\documentclass[a4paper,UKenglish]{lipics-v2018}

\pdfoutput=1

\title{Strong link between BWT and XBW via Aho-Corasick automaton and applications to Run-Length Encoding}

\titlerunning{From \BWT to \XBW, via Aho-Corasick, and back}

\author{Bastien Cazaux}{Department of Computer Science, University of Helsinki, Helsinki, Finland\\ L.I.R.M.M., CNRS, Universit\'e Montpellier, Montpellier, France}{bastien.cazaux@cs.helsinki.fi}{}{}

\author{Eric Rivals}{L.I.R.M.M., CNRS, Universit\'e Montpellier, Montpellier, France\\ Institute of Computational Biology, Montpellier, France}{rivals@lirmm.fr}{}{}

\authorrunning{B. Cazaux and E. Rivals}

\Copyright{Bastien Cazaux and Eric Rivals}

\subjclass{}


\category{}

\relatedversion{}

\supplement{}

\funding{}

\acknowledgements{}

\EventEditors{John Q. Open and Joan R. Access}
\EventNoEds{2}
\EventLongTitle{42nd Conference on Very Important Topics (CVIT 2016)}
\EventShortTitle{CVIT 2016}
\EventAcronym{CVIT}
\EventYear{2016}
\EventDate{December 24--27, 2016}
\EventLocation{Little Whinging, United Kingdom}
\EventLogo{}
\SeriesVolume{42}
\ArticleNo{23}

\usepackage{tikz}
\usepackage{xspace,xcolor}
\usepackage[ruled,vlined]{algorithm2e}

\newcommand{\eg}{\emph{e.g.}}
\newcommand{\ie}{i.e.\xspace}

\newcommand{\etal}{\textit{et al.}\xspace}
\newcommand{\iif}{iff\xspace}

\newcommand{\taille}[1]{\vert #1 \vert}
\newcommand{\norme}[1]{\left\lVert #1 \right\rVert}

\newcommand{\acro}[1]{\ensuremath{\mathtt{#1}}\xspace}

\newcommand{\BWT}{\acro{BWT}}
\newcommand{\BWTs}{\acro{BWTs}}
\newcommand{\LCP}{\acro{LCP}}

\newcommand{\AC}{\acro{AC}}
\newcommand{\ACT}{\acro{ACT}}
\newcommand{\ACFL}{\acro{ACFL}}

\newcommand{\SA}{\acro{SA}}

\newcommand{\Rank}{\acro{rank}}
\newcommand{\Select}{\acro{select}}

\newcommand{\LF}{\acro{LF}}
\newcommand{\LRS}{\acro{LRS}}

\newcommand{\PA}{\acro{PA}}

\newcommand{\Decomp}{\acro{Decomp}\_{\BWT}}
\newcommand{\Decompb}{\acro{Decomp}\_{\XBW}}

\newcommand{\XBW}{\acro{XBW}}
\newcommand{\XBWT}{\acro{XBWT}}
\newcommand{\XBWL}{\acro{XBWL}}

\newcommand{\Prefix}{\acro{Prefix}}
\newcommand{\Suffix}{\acro{Suffix}}

\newcommand{\inter}[2]{[\![#1,#2]\!]}

\newcommand{\mpb}{\acro{Min}-\acro{Permutation}-\acro{BWT}}
\newcommand{\mpx}{\acro{Min}-\acro{Permutation}-\acro{XBWT}}
\newcommand{\mpt}{\acro{Min}-\acro{Permutation}-\acro{Table}}

\newcommand{\BWD}{\acro{BWD}}
\newcommand{\XBWD}{\acro{XBWD}}

\nolinenumbers

\begin{document}
\maketitle

\begin{abstract}
The boom of genomic sequencing makes compression of set of sequences inescapable. This underlies the need for multi-string indexing data structures that helps compressing the data. The most prominent example of such data structures is the Burrows-Wheeler Transform (BWT), a reversible permutation of a text that improves its compressibility. A similar data structure, the eXtended Burrows-Wheeler Transform (XBW), is able to index a tree labelled with alphabet symbols. 
A link between a multi-string BWT and the Aho-Corasick automaton has already been found and led to a way to build a XBW from a multi-string BWT. We exhibit a stronger link between a multi-string BWT and a XBW by using the order of the concatenation in the multi-string. This bijective link has several applications: first, it allows to build one data structure from the other; second, it enables one to compute an ordering of the input strings that optimises a Run-Length measure (i.e., the compressibility) of the BWT or of the XBW.

  
\end{abstract}

\section{Introduction}
 A seminal, key data structure, which was used for searching a set of words in a text, is the Aho-Corasick (AC) automaton \cite{Aho1975} Aho-Corasick automaton. Its states form a tree that indexes all the prefixes of the words, and each node in the tree is equipped with another kind of arc, called a Failure Link. A failure link of a node/prefix $v$ points to the node representing the largest proper suffix of $v$ in the tree. In a way, the Aho-Corasick automaton can be viewed as a multi-string indexing data structure.

 In the early 1990, the Burrows-Wheeler Transform (\BWT) of a text $T$, which is a reversible permutation of $T$, was introduced for the sake of compressing a text. Indeed, the \BWT permutation tends to groups identical symbols in runs, which favours compression~\cite{Burrows94}. However, the \BWT can also be used as an index for searching in $T$, using the Backward Search procedure \cite{FerraginaM05}. In fact, the \BWT of $T$ is the last column of a matrix containing all cyclic-shifts of $T$ sorted in lexicographical order. As sorting the cyclic shifts of $T$ is equivalent to sorting its suffixes, there exists a natural link between the Suffix arrays of $T$ and the \BWT of $T$.
 Starting in 2005, the radical increase in textual data and in biological sequencing data raise the need for multi-string indexes. In the multi-string case, all input strings are concatenated, separated by a termination symbol (which does not belong to the alphabet), and then indexed in a traditional indexing data structure (\eg, a suffix array or FM-index). Such multi-string indexes are heavily exploited in bioinformatics: first, to index all chromosomes of a genome \cite{bwa}, or a large collection of similar genomes, which allows aligning sequencing reads simultaneously to several reference genomes \cite{MakinenNSV09}, or second, to store and mine whole sets of raw DNA/RNA sequencing reads for the purpose of comparing biological conditions or of identifying splice junctions in RNA, etc \cite{Cox-etal-12,HoltM14}. In fact, managing compressed and searchable read data sets is now crucial for bioinformatic analyses of such data.
 
 Initially viewed as a simple extension of single-string \BWT construction, the efficient construction of multi-string \BWT is not trivial and has been investigated per se. Bauer \etal proposed a lightweight incremental algorithm for their construction \cite{Bauer-etal-13} among others. Then, Holt \etal devised algorithms for directly merging several, already built multi-string \BWTs efficiently \cite{HoltM14,HoltM14b}, which has been recently improved to simultaneously build the companion \emph{Longest Common Prefix} (\LCP) table \cite{EgidiM17} or to scale up to terabyte datasets \cite{Siren16}.
 
The notion of \BWT has been extended into the \XBW to index trees whose arcs are labelled by alphabet symbols~\cite{FerraginaLMM09}. The \XBW, takes the form two arrays, which compactly represent the tree and offer navigational operations.

Recently, Gagie et al. propose the notion of Wheeler graphs to subsume several variants of the \BWT, including the \XBW of a trie for a set of strings~\cite{GagieMS17}.
The relation between a multi-string \BWT and the \XBW representation of the Aho-Corasick automaton has already been studied and exploited. Hon et al. first use the \XBW representation of the Aho-Corasick trie to speed-up dictionary matching~\cite{HonKSTV13} building up on~\cite{Belazzougui10}. Manzini gave an algorithm that computes the failure links for the trie using the Suffix Array and \LCP tables, and an algorithm to build the \XBW of the trie with failure links from the multi-string \BWT \cite{Manzini16}. However, none these establish a bijective link between a multi-string \BWT and the \XBW of the Aho-Corasick automaton. To generalize these results, one need to consider the order in which the strings are concatenated for form the multi-string. This idea enables us to exhibit a bijection between a multi-string \BWT and the \XBW of the Aho-Corasick automaton, which allows building one structure from the other in either direction (from \BWT to \XBW or from \XBW to \BWT). Finally, we exploit this bijection between the \BWT and the \XBW to find an optimal string order that maximises a Run-Length Encoding (\ie, the compressibility) of these two data structures.

\section{Notation}

Let $i$ and $j$ be two integers. The \emph{interval} $\inter{i}{j}$ is the set of all integers between $i$ and $j$. An \emph{integer interval partition} of $\inter{i}{j}$ is a set of intervals $\{ \inter{i_1}{j_1}, \ldots \inter{i_n}{j_n}\}$ such that $i_1 = i$, $j_n = j$, and for all $k \in \inter{1}{n-1}$, $j_k+1 = i_{k+1}$. We define also the order $<$ on intervals such that for two intervals $u = \inter{i}{j}$ and $v = \inter{i'}{j'}$, $u < v$ \iif $j < i'$.
Let $E$ be a finite set and let $\#(E)$ denote its cardinality.
A \emph{permutation} of  $E$ is an automorphism of $E$. A permutation $\sigma$ of $E$ is said to be \emph{circular} \iif for all $i$ and $j \in E$, there exists a positive integer $k$ such that $\sigma^k (i) = j$. For a circular permutation $\sigma$ of $E$ and an element $e\in E$, we denote by $\overline{\sigma_e}$ the function from $\inter{1}{\#(E)}$ to $E$ such that for all $i \in \inter{1}{\#(E)}$, $\overline{\sigma_e}(i) = {\sigma}^{(i-1)}(e)$. If $E$ is totally ordered by $<$, we define $<_{\sigma}$ the order on $E$ such that $e <_{\sigma} f \text{ \textbf{\iif} } \sigma(e) < \sigma(f)$ for any $e, f \in E$.

Let $\Sigma$ be a finite alphabet. A \emph{string} $w$ of length $n$ over $\Sigma$ is a sequence of symbols $w[1]\ldots w[n]$ where $w[i] \in \Sigma$ for all $i \in \inter{1}{n}$. $\Sigma^{\ast}$ is the set of all the strings of $\Sigma$.
The \emph{length} of a string $w$ is denoted by $\taille{w}$. A \emph{substring} of $w$ is written as $w[i,j] = w[i]\ldots w[j]$. A \emph{prefix} of $w$ is a substring which begins $w$ and a \emph{suffix} of $w$ is a substring which ends $w$. 
The \emph{reverse} of a string $w$, denoted by $\overleftarrow{w}$, is the string $w[n]\ldots w[2]w[1]$. We define the \emph{lexicographic order} $<$ on strings as usual.

Let $S$ be a set of strings. The norm of $S$, denoted $\norme{S}$, is the sum of the length of strings of $S$. Let $\Prefix(S)$, (respectively $\Suffix(S)$) denote the set of all prefixes (resp. all suffixes) of strings of $S$. We denote by $\overleftarrow{S}$ the set of all reverse strings of strings of $S$.

An \emph{ordered set of strings} $P$ is a pair $(S,\sigma)$ where $S$ is a set strings in lexicographic order, and $\sigma$ a circular permutation of $S$. We denote by $P.S$ the set of strings $S$ and by $P.\sigma$ the circular permutation $\sigma$. We denote by $\overleftarrow{P}$ the pair $(\overleftarrow{P.S},P.\sigma)$. 

Let $\mathcal{T}$ be a tree and $u$ be a node of $\mathcal{T}$. Let  $\perp$ denote the root of $\mathcal{T}$. We denote by $\acro{Parent}_{\mathcal{T}} (u)$ the parent of $u$ in $\mathcal{T}$, by $\acro{Children}_{\mathcal{T}} (u)$ the set of children of $u$ in $\mathcal{T}$, and by $\acro{Leaves}_{\mathcal{T}} (u)$ the set of leaves in the subtree of $u$ in $\mathcal{T}$. Let $v$ be a leaf of $\mathcal{T}$; we denote by $\mathcal{T}(v)$ the subtree of $\mathcal{T}$ containing all nodes comprised between $\perp$ and $v$ included. As for a leaf $v$ in the subtree of $u$ in $\mathcal{T}$, $\#(\acro{Children}_{\mathcal{T}(v)} (u)) = 1$, we denote by $\acro{Child}_{\mathcal{T}(v)} (u)$ the unique element of $\acro{Children}_{\mathcal{T}(v)} (u)$.
Let $\prec$ be a total order on $\acro{Leaves}_{\mathcal{T}} (\perp)$. Then,  for any node $v$ of $\mathcal{T}$, $\prec$ also is a total order on $\acro{Leaves}_{\mathcal{T}} (v)$. We extend $\prec$ to the set $\acro{Children}_{\mathcal{T}} (v)$ for any node $v$ of $\mathcal{T}$ as follows: for any $x,y$ in $\acro{Children}_{\mathcal{T}} (v)$, $x \prec y \text{ \textbf{iif} } \min_{x' \in \acro{Leaves}_{\mathcal{T}} (x)} x' \prec  \min_{y' \in \acro{Leaves}_{\mathcal{T}} (y)} y'$.

\begin{figure}[ht]
  \centering
    \includegraphics[scale=.91]{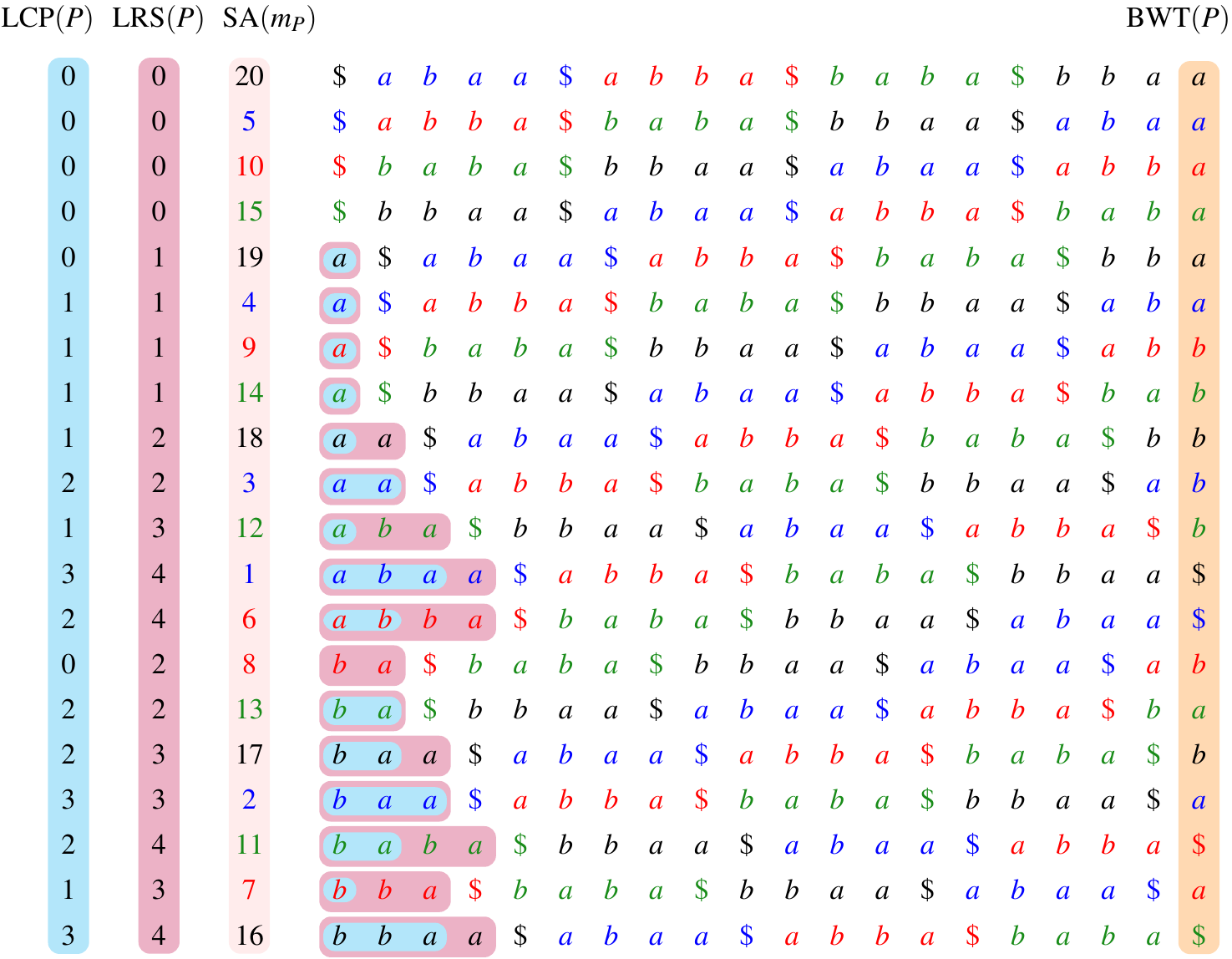}
  \caption{Tables $\LCP(P)$, $\LRS(P)$, $\SA(m_P)$ and $\BWT(P)$ for the running example. {One has $\LCP(P)[10] = \LRS(P)[10] = 2$ although the longest common prefix between suffixes of rank $9$ and $10$ is $aa\$ab$ of length $5$ (\ie, although $\LCP(m_P)[10] = 5$).}
\label{fig:bwt}}
\end{figure}

\section{Decomposition of the \BWT and link with Aho-Corasick automaton}\label{sec:ac}

Here, we introduce a decomposition of a multi-string \BWT that leads to exhibit a bijection with the Aho-Corasick automaton. This builds on and extends Manzini's work~\cite{Manzini16}.

\subsection{Decomposition of the \BWT's positions}

\subparagraph{BWT of a string}
Let $w$ be a string and $i$ be an integer satisfying $1 \leq i \leq \taille{w}$.
The \emph{Suffix Array} (\SA) of $w$~\cite{ManberM93}, denoted $\SA(w)$, is the array of integers that stores the starting positions of the $\taille{w}$ suffixes of $w$ sorted in lexicographic order. The \emph{Burrows-Wheeler Transform} (\BWT)~\cite{Burrows94} of $w$, denoted $\BWT(w)$, is the array containing a permutation of the symbols of $w$ which satisfies $\BWT(w)[i] = w[\SA(w)[i]-1]$ if $\SA(w)[i] > 1$, and $\BWT(w)[i] = w[\taille{w}]$ otherwise. The \emph{Longest Common Prefix} table (\LCP)~\cite{Burrows94} of $w$, denoted $\LCP(w)$, is the array of integers such that $\LCP(w)[i]$ equals the length of the longest common prefix between the suffixes of $w$ starting at positions $\SA(w)[i]$ and $\SA(w)[i-1]$ if $i>1$, and $0$ otherwise.

For any string $s$ and $c \in \Sigma$, one defines the functions denoted \Rank and \Select as follows: $\Rank_c(s,i)$ is the number of occurences of $c$ in $s[1,i]$, and $\Select_c(s,j)$ is the position of the $j^{th}$ occurence of $c$ in $s$. 
The arrays $\BWT(w)$ and $\LCP(w)$ can be computed in $O(\taille{w})$ time~\cite{Burrows94,FerraginaM05}. Simultaneously, one can compute \Rank and \Select for $\BWT(w)$ at no additional cost and implement them such that any \Rank or \Select query takes constant time~\cite{GrossiGV03,FoschiniGGV06,MakinenN07}. We use such state-of-the-art structure to store a \BWT.

Let $C$ denote the array of length $\Sigma$ such that $C[c]$ equals the number of symbols of $w$ that are alphabetically strictly smaller than $c$.
The \emph{Last-to-First column mapping} (\LF)~\cite{FerraginaM05} of $w$ is the function such that for any $1 \leq i \leq \taille{w}$ one has $\LF(w)[i] = C[\BWT(w)[i]] + \Rank_{\BWT(w)[i]}(\BWT(w),i)$. It is proven that $\SA(w)[\LF(w)[i]] = \SA(w)[i]-1$ if $\SA(w)[i] \geq 2$ and $\SA(w)[\LF(w)[i]] = \taille{w}$ otherwise~\cite{Burrows94,FerraginaM05}.

\subparagraph{BWT of a set of strings}
\textbf{From now on}, let $P$ be an ordered set of strings. We assume the symbol $\$$ is not in $\Sigma$ and is alphabetically smaller than all other symbols.  We denote by $m_P$ the string obtained by  concatenating the strings of $P.S$ separated by a $\$$ and following the order $\overline{P.\sigma_{s_1}}$. I.e., $ m_P := s_{\overline{P.\sigma_{s_1}}(1)} \$ s_{\overline{P.\sigma_{s_1}}(2)} \$ \ldots \$ s_{\overline{P.\sigma_{s_1}}(n)} \$ $ (See Figure~\ref{fig:mp}).

\begin{figure}[b]
  \centering
    \includegraphics[scale=1]{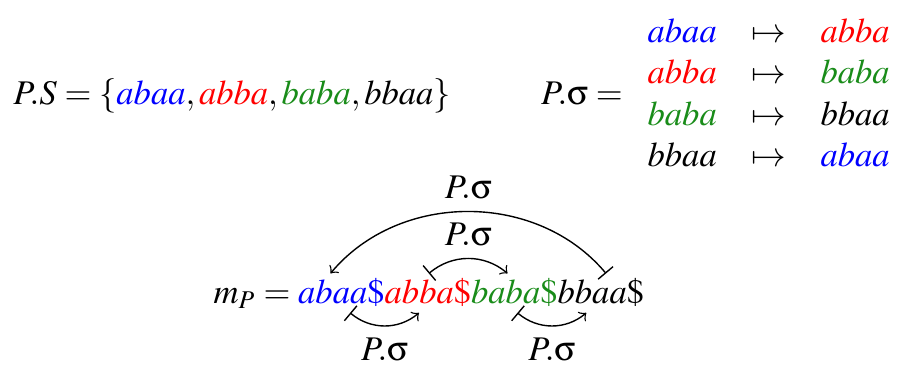}
  \caption{Runnig example with $P.S = \{abaa,abba,baba,bbaa\}$ and the corresponding $m_P$.\label{fig:mp}}
\end{figure}

We extend the notion of \BWT of a string to an ordered set of strings $P$: the \BWT of $P$ is the \BWT of the string $m_P$, i.e. $\BWT(P) = \BWT(m_P)$. We extend similarly the \LF function by setting that $\LF(P) = \LF(m_P)$.

{We define the \emph{Longest Representative Suffix} table (\LRS) of $P$ as the array of $\taille{m_P}$ integers satisfying: for any $i \in \inter{1}{\taille{m_P}}$ one has $\LRS(P)[i] = \Select_{\$}(m_P[\SA(m_P)[i]:\taille{m_P}],1) - 1$.} The entry $\LRS(P)[i]$ gives the length of the substring of $m_P$ starting at position $\SA(m_P)[i]$ up to the next $\$ $ not included. Using the \LRS table, we extend the notion of \LCP table to an ordered set of strings. For $P$, we set $\LCP(P)[i] = \min(\LCP(m_P)[i],\LRS(P)[i])$ for any $i \in \inter{1}{\taille{m_P}}$ (see Figure~\ref{fig:bwt}).

We get Lemma~\ref{le:lrs} and Proposition~\ref{prop:linear:lrs}, whose proofs can be found in Appendix.
\begin{lemma}\label{le:lrs}
  Let be $i \in \inter{1}{\taille{m_P}}$. 
  \[
    \LRS(P)[i] \quad = \quad
    \begin{cases}
      0 & \text{if } i \in \inter{1}{\#(P.S)},\\
      \LRS(P)[\LF(P)[i]]-1 & \text{otherwise}.
    \end{cases}
  \]
\end{lemma}

\begin{proposition}\label{prop:linear:lrs}
  Let $P$ be an ordered set of strings. Using tables $\BWT(P)$ and $\LCP(m_P)$, we can compute the tables $\LRS(P)$ and $\LCP(P)$ in linear time in $\taille{m_P}$.
\end{proposition}

\subparagraph{Decomposition of a multi-string \BWT}
Let $P$ be an ordered set of strings. Let $\Decomp(P)$ be the integer interval partition of $\inter{1}{\taille{m_P}}$ such that
\[
  \inter{i}{j} \in \Decomp(P)  \quad \text{ \textbf{\iif} } \quad
  \begin{cases}
    \LCP(P)[k] \neq \LRS(P)[k], & \text{ for } k \in \{i,j+1\} \\
    \LCP(P)[k] = \LRS(P)[k], & \text{ for } k \in \inter{i+1}{j}.
  \end{cases}
\]
We define $\acro{Dec\_Pre}$ the function from $\Decomp(P)$ to $\Sigma^{\ast}$ such that for $u =\inter{i}{j} \in \Decomp(P)$, 
\[
  \acro{Dec\_Pre}[u]  \; = \; \overleftarrow{m_P[\SA(m_P)[i]:\SA(m_P)[i]+\LRS(P)[i]-1]}.
\]

\begin{proposition}\label{prop:bij:decomp:prefix}
  $\acro{Dec\_Pre}$ is a bijection between $\Decomp(P)$ and $\Prefix(\overleftarrow{P.S})$.
\end{proposition}

\subsection{Link between \BWT and \AC}
\label{sec:bwt:ac}

\subparagraph{Aho-Corasick tree for a set of strings}
{The \emph{Aho-Corasick automaton} (\AC)~\cite{Aho1975} of a set of strings $S$ is a digraph whose set of nodes is the set of all prefixes of the strings of $S$. This graph is composed of two trees on the same node set. The first tree, which we called the \emph{Aho-Corasick Tree} (\ACT), has an arc from a prefix $u$ to a different prefix $v$ \iif $u$ is the longest prefix of $v$ among $\Prefix(S)$ (see Figure~\ref{fig:link:bwt:ac}). The second tree, termed \emph{Aho-Corasick Failure link} (\ACFL), has an arc from a prefix $u$ to a different prefix $v$ \iif $v$ is the longest suffix of $u$ among $\Prefix(S)$.}


By Proposition~\ref{prop:bij:decomp:prefix}, there exists an integer interval partition of $\inter{1}{\taille{m_P}}$ (i.e. \Decomp(P)) that is in bijection with the set of nodes of $\AC(\overleftarrow{P.S})$.

\begin{proposition}[See Figure~\ref{fig:link:bwt:ac}]\label{prop:iso:act}
  The graph $G_T(P) = (\Decomp(P),A_T(P))$ is isomorphic to the tree $\ACT(\overleftarrow{P.S})$, where
  \[
    A_T(P) := \{(u,v) \in \Decomp(P)^2 |\ \exists x \in u \text{ such that } \LF(P)[x] \neq 0 \text{ and } \LF(P)[x] \in v \}.
  \]
\end{proposition}

%
\begin{proposition}[See Figure~\ref{fig:link:bwt:ac}]\label{prop:iso:acfl}
  The graph $G_F(P) = (\Decomp(P),A_F(P))$ is isomorphic to the tree $\ACFL(\overleftarrow{P.S})$, where
  \[
    A_F(P) := \{(u,v) \in \Decomp(P)^2 |\ \big(\max_{\substack{k<i \\ \LRS(P)[k] = \min_{k-1 \leq l \leq i}(\LCP(P)[l])}}( k )\big) \in v \text{ with } u = \inter{i}{j}\}.
  \]
\end{proposition}
%

Finally, next theorem states how to simulate an Aho-Corasick automaton using the \BWT (as in~\cite{Manzini16}).

\begin{theorem}[See Figure~\ref{fig:link:bwt:ac}]\label{th:simulate:ac:bwt}
   Using tables $\BWT(\overleftarrow{P})$, $\LCP(\overleftarrow{P})$, $\LRS(\overleftarrow{P})$ and the functions $\LF(\overleftarrow{P})$, $\Rank$ and $\Select$, we can build a graph that is isomorphic to $\AC(P.S)$.
\end{theorem}
\begin{proof}
  Taking the set of strings $Q = \overleftarrow{P}$, Proposition~\ref{prop:iso:act} implies that the graph $G_T(Q) = G_T(\overleftarrow{P})$ is isomorphic to the tree $\ACT(\overleftarrow{Q.S}) = \ACT(\overleftarrow{\overleftarrow{P.S}}) = \ACT(P.S)$. Proposition~\ref{prop:iso:acfl} says that the graph $G_F(Q) = G_F(\overleftarrow{P})$ is isomorphic to the tree $\ACFL(\overleftarrow{Q.S}) = \ACFL(\overleftarrow{\overleftarrow{P.S}}) = \ACFL(P.S)$.
\end{proof}
\begin{figure}
  \centering
  \includegraphics[scale=0.7]{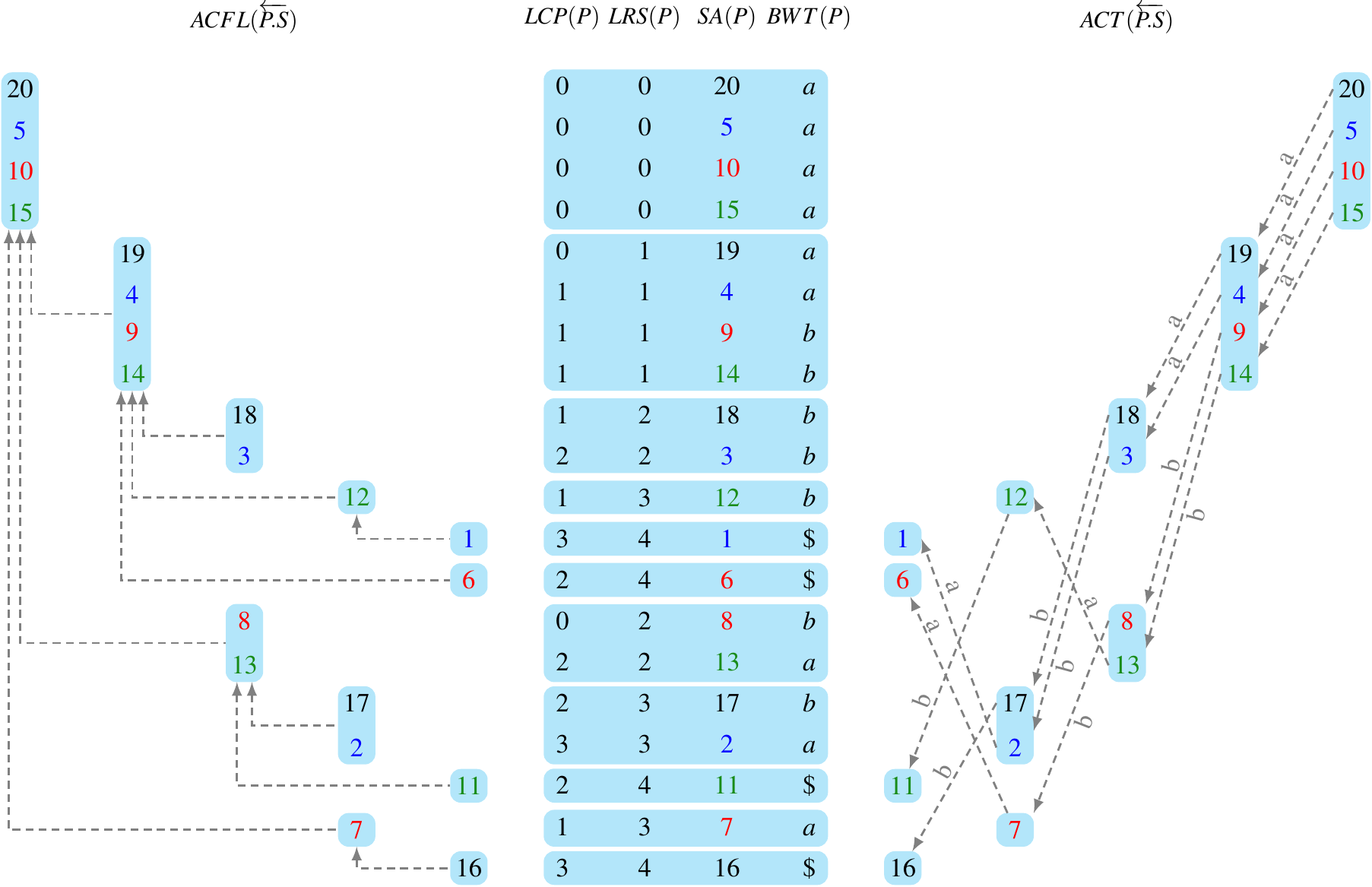}
  \caption{Link between $\BWT(P)$, $\ACT(\overleftarrow{P.S})$ and $\ACFL(\overleftarrow{P.S})$ for the running example.\label{fig:link:bwt:ac}}
\end{figure}

\begin{figure}[!ht]
  \centering
    \includegraphics[scale=0.85]{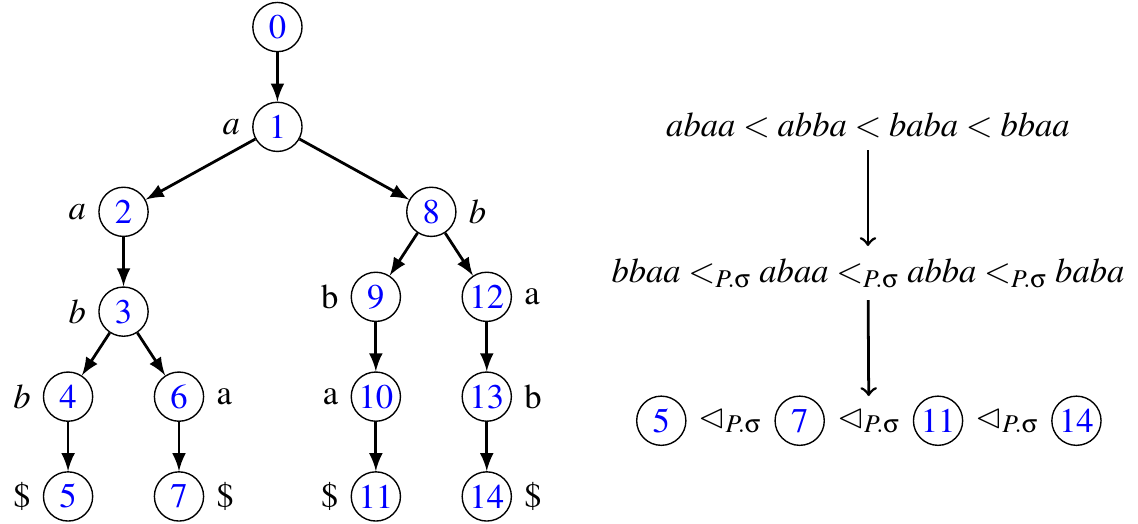}
  \caption{Tree $\ACT(\overleftarrow{P})$ for the running example, and links between the orders $<$, $<_{P.\sigma}$, and $\lhd_{P.\sigma}$.\label{fig:xbw}}
\end{figure}

\section{Link between \BWT and \XBW}
\label{sec:bwt:xbw}

In Section~\ref{sec:ac}, we gave a new proof of the relation between the Aho-Corasick automaton and the \BWT. Here, we exhibit a new (bijective) link between the \BWT and the \XBW, which takes into account the order in the multi-string  (Theorem~\ref{th:link:bwt:xbw}). This leads to both, another construction algorithm of the \XBW from the \BWT, and to a construction of the \BWT from the \XBW, and thereby extends Manzini's results (Corollary~\ref{cor:build:bwt:xbw}).

\subparagraph{XBW of a tree}
\newcommand{\ntreize}{\tikz \node[draw,circle,scale=.465]{13};}
\newcommand{\nquatre}{\tikz \node[draw,circle,scale=.465]{4};}

Let $\mathcal{T}$ be an ordered tree such that every node of $\mathcal{T}$ is labelled with a symbol from an alphabet $\Sigma$. We define the functions $\delta$ and $\pi$ on the set of nodes of $\mathcal{T}$ such that for a node $v$ of $\mathcal{T}$, $\delta(v)$ is the label of $v$, and $\pi(v)$ is the string obtained by concatenating the labels from $v$'s parent to the root of $\mathcal{T}$. Let $\prec$ be the total order between the nodes of $\mathcal{T}$ such that for $u$ and $v$ two nodes of $\mathcal{T}$, $u \prec v$ \iif $\pi(u)$ is strictly lexicographically smaller than $\pi(v)$ or $u$ is before $v$ in the order of  $\mathcal{T}$. \textbf{Example}:  With the tree of Figure~\ref{fig:xbw}, on the nodes numbered $13$ and $4$, we have $\delta(13) = b$ and $\Pi(13) = aba = \delta(12)\delta(8)\delta(1)$, and also $\delta(4) = b$ and $\Pi(4) = baa = \delta(3)\delta(2)\delta(1)$. Thus, $\ntreize \prec \nquatre$.

The \emph{Prefix Array} (\PA) of an ordered tree $\mathcal{T}$ is the array of pointers to the nodes of $\mathcal{T}$ (except the root of $\mathcal{T}$) sorted in $\prec$ order.
The \emph{eXtended Burrows-Wheeler Transform} (\XBWT)~\cite{FerraginaLMM09}\footnote{In~\cite{FerraginaLMM09}, the \emph{\XBW-transform} is defined as $\XBW(T)[i] =\langle\XBWT[i],\XBWL[i]\rangle$, for any position $i$.} of a tree $\mathcal{T}$ is an array of symbols of $\Sigma$, of length $\PA(\mathcal{T})$ such that the entry at position $i$ gives the label of the node $\PA(\mathcal{T})[i]$.
The \emph{eXtended Burrows-Wheeler Last} (\XBWL)~\cite{FerraginaLMM09}\footnotemark[\value{footnote}] of a tree $\mathcal{T}$ is the bit array of length of $\PA(\mathcal{T})$ such that $\PA(\mathcal{T})[i]$ equals $1$ if the node $\PA(\mathcal{T})[i]$ is a last child of its parent, and $0$ otherwise.

Similarly to the definition of $\Decomp(P)$ for an ordered set of strings $P$, we define $\Decompb(\mathcal{T})$ for a tree $\mathcal{T}$ of $t$ nodes as the integer interval partition of $\inter{1}{t}$ such that
\[
  \inter{i}{j} \in \Decompb(\mathcal{T}) \; \text{ \textbf{\iif} } \;
  \begin{cases}
    \XBWL(\mathcal{T})[k] = 1, & \text{ for } k \in \{i-1,j\} \\
    \XBWL(\mathcal{T})[k] = 0, & \text{ for } k \in \inter{i}{j-1}.
  \end{cases}
\]

\subparagraph{XBW of an Aho-Corasick tree}

For a set of strings $S = \{s_1,\ldots,s_n\}$, we denote by $S^{\$}$ the set $\{s_1\$,\ldots,s_n\$\}$.
Let $P$ be an ordered set of strings. We define $\ACT(P)$ as the Aho-Corasick tree of $P.S^{\$}$ equipped with the order $\lhd_{P.\sigma}$.  Indeed, $\lhd_{P.\sigma}$ is the order on the leaves satisfying: for $u$ and $v$ two leaves of $\ACT(P)$, $u \lhd_{P.\sigma} v$ \iif $\pi(u) <_{P.\sigma} \pi(v)$. We extend this order to the set of children of all nodes (See Figure~\ref{fig:xbw}). Note that $\ACT(P)$ differs from $\ACT(P.S)$, which was defined in Section~\ref{sec:bwt:ac}.

\begin{theorem}[See Figures~\ref{fig:bwt:xbw} and~\ref{fig:bwt:xbw:tree}]\label{th:link:bwt:xbw}
  There exists a bijection \acro{BWT\_XBW} between $\Decomp(P)$ and $\Decompb(\ACT(\overleftarrow{P}))$ such that for all $u \in \Decomp(P)$ with $u = \inter{i}{j}$, $\acro{BWT\_XBW}(u) = \inter{i'}{j'}$ and $z = \acro{Parent}_{\ACT(\overleftarrow{P})} (\PA(\ACT(\overleftarrow{P})[i']))$.
  \begin{itemize}
  \item Let $\{y_1,y_2,\ldots y_{\#(u)}\} = \acro{Leaves}_{\ACT(\overleftarrow{P})} (z)$ such that $y_i \lhd_{P.\sigma} y_j \Rightarrow i < j$; then
    \[
      \BWT(P)[i,j] = \delta[\acro{Child}_{\ACT(\overleftarrow{P})(y_1)} (z)]\delta[\acro{Child}_{\ACT(\overleftarrow{P})(y_2)} (z)]\ldots \delta[\acro{Child}_{\ACT(\overleftarrow{P})(y_{\#(u)})} (z)].
    \]
  \item Let $\{x_1,x_2,\ldots, x_{\#(\acro{BWT\_XBW}(u))}\} = \acro{Children}_{\ACT(\overleftarrow{P})} (z)$ such that $x_i \lhd_{P.\sigma} x_j \Rightarrow i < j$; then
    \[
      \XBWT(\ACT(\overleftarrow{P}))[i',j'] = \delta[x_1]\delta[x_2]\ldots \delta[x_{\#(\acro{BWT}\_\acro{XBW}(u))}].
    \]
  \end{itemize}
\end{theorem}

\begin{figure}
  \centering
    \includegraphics[scale=0.85]{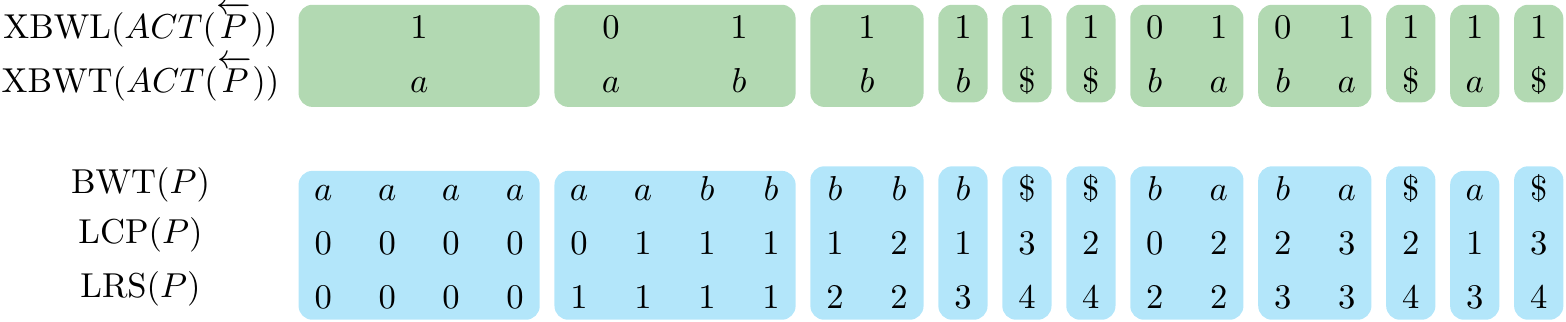}
  \caption{Link between $\Decomp(P)$ (in blue) and $\Decompb(\ACT(\overleftarrow{P}))$ (in green).\label{fig:bwt:xbw}}
\end{figure}

\begin{figure}
  \centering
    \includegraphics[scale=0.85]{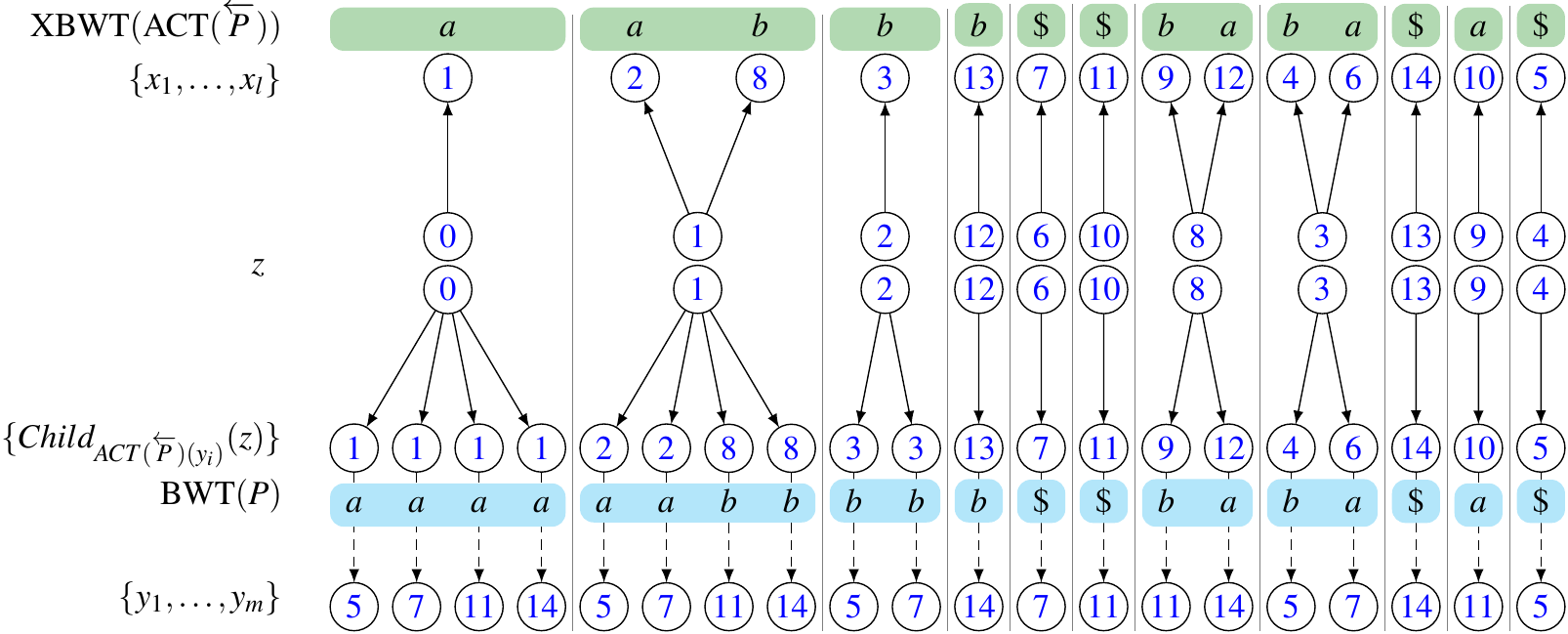}
  \caption{Illustration of Theorem~\ref{th:link:bwt:xbw} for the running example.\label{fig:bwt:xbw:tree}}
\end{figure}

\begin{proof}
  We define $T_B$ (resp. $T_X$) as the array of intervals of $\Decomp(P)$ \\(resp. $\Decompb(\ACT(\overleftarrow{P}))$) sorted in the interval order. 
  Let us prove that $T_B$ and $T_X$ have the same length, and that at the same position $i$, $T_B[i]$ and $T_X[i]$ represent the same prefix of $\overleftarrow{P}$.
  By Proposition~\ref{prop:bij:decomp:prefix}, the length of $T_B$ is $\#(\Prefix(\overleftarrow{P.S}))$. By the definition of $\Decompb$, the length of $T_X$ is the number of $1$ in $\XBWL(\ACT(\overleftarrow{P}))$, \ie the number of internal nodes of $\ACT(\overleftarrow{P})$, and thus is equal to $\#(\Prefix(\overleftarrow{P.S}))$. Hence, $T_B$ and $T_X$ have the same number of elements.

  Let be $i \in \inter{1}{\#(\Prefix(\overleftarrow{P.S}))}$. By Proposition~\ref{prop:bij:decomp:prefix}, $T_B[i]$ represents the $i^{th}$ suffix of strings of $P.S$ in lexicographic order. By the definition of $\XBWT(\ACT(\overleftarrow{P}))$, all the nodes $\PA(\ACT(\overleftarrow{P}))[k]$ for $k \in T_X[i]$ have the same parent $z$ in $\ACT(\overleftarrow{P})$, and $z$ represents the $i^{th}$ nodes in $\prec$ order, \ie the $i^{th}$ suffix of strings of $P.S$ in lexicographic order. 

  As for a given position $i$, $T_B[i]$ and $T_X[i]$ represent the same prefix of $\overleftarrow{P}$, we define the bijection $\acro{BWT\_XBW}$ such that for all $i$ in $\inter{1}{\#(\Prefix(\overleftarrow{P.S}))}$, $\acro{BWT}\_\acro{XBW}[T_B[i]] = T_X[i]$. 
As the tree $\ACT(\overleftarrow{P})$ represents the Aho-Corasick tree of $\overleftarrow{P.S}^{\$}$, we have a bijection $b_1$ from the node set of $\ACT(\overleftarrow{P})$ onto the set of prefixes of $\overleftarrow{P.S}^{\$}$. By the definition of functions $\pi$ and $\delta$, for any node $v$ of $\ACT(\overleftarrow{P})$, we have $b_1(v) = \overleftarrow{\pi(v)}\; \delta(v)$.

By the definition of $\XBWT(\ACT(\overleftarrow{P}))$, we have that for $T_X[i] = \inter{i'}{j'}$
\[
  \begin{array}{rcl}
    \#(T_X[i]) & = & \#\{ v \text{ node of } \ACT(\overleftarrow{P}) \ | \ \pi(v) = \overleftarrow{\acro{Dec\_Pre}[\acro{BWT\_XBW}^{-1}[T_X[i]]]} \}\\ 
               & = & \#\{ v \text{ node of } \ACT(\overleftarrow{P}) \ | \ v \text{ is a child of }b_1^{-1}(\acro{Dec\_Pre}[\acro{BWT\_XBW}^{-1}[T_X[i]]])\}. 
  \end{array}
\]
As $i' \in T_X[i]$ and $b_1^{-1}(\acro{Dec\_Pre}[\acro{BWT\_XBW}^{-1}[T_X[i]]]) = \acro{Parent}_{\ACT(\overleftarrow{P})} (\PA(\ACT(\overleftarrow{P}))[i'])$, 
\[
  \begin{array}{rcl}
    \#(T_X[i]) & = & \#\{ v \text{ node of } \ACT(\overleftarrow{P}) \ | \ v \text{ is a child of }\acro{Parent}_{\ACT(\overleftarrow{P})} (\PA(\ACT(\overleftarrow{P}))[i'])\}\\ 
               & =  & \#(\acro{Children}_{\ACT(\overleftarrow{P})} (\acro{Parent}_{\ACT(\overleftarrow{P})} (\PA(\ACT(\overleftarrow{P}))[i']))).
  \end{array}
\]
Let $\{x_1,x_2,\ldots, x_{\#(T_X[i])}\}$ be the set of children of $\acro{Parent}_{\ACT(\overleftarrow{P})} (\PA(\ACT(\overleftarrow{P}))[i'])$ sorted such that $x_1 \lhd_{P.\sigma} \ldots \lhd_{P.\sigma} x_{\#(T_X[i])}$.
By the definition of $\XBWT(\ACT(\overleftarrow{P}))$, for any $k \in \inter{1}{\#(T_X[i])}$, we have $\XBWT(\ACT(\overleftarrow{P}))[i'+k-1] = \delta[x_k]$.
By the definition of $\BWT(P)$, for $T_B[i] = \inter{i''}{j''}$, we get
  \[
    \begin{array}{rcl}
      \#(T_B[i]) & = & \#(\{w \in \overleftarrow{P.S}^{\$} \ | \   \acro{Dec\_Pre}[T_B[i]] \text{ is a prefix of } w \}) \\
      & = & \#(\{v \text{ leaf of } \ACT(\overleftarrow{P}) \ | \ b_1^{-1}(\acro{Dec\_Pre}[T_B[i]]) \text{ is an ancestor of } v \text{ in } \ACT(\overleftarrow{P}) \}).
    \end{array}
  \]
As $b_1^{-1}(\acro{Dec\_Pre}[T_B[i]]) = b_1^{-1}(\acro{Dec\_Pre}[\acro{BWT\_XBW}^{-1}[T_X[i]]]) = \acro{Parent}_{\ACT(\overleftarrow{P})} (\PA(\ACT(\overleftarrow{P}))[i'])$,
\[
\begin{array}{rcl}
\#(T_B[i]) &= &\#\{ v \text{ a leaf of } \ACT(\overleftarrow{P}) \ | \ \acro{Parent}_{\ACT(\overleftarrow{P})} (\PA(\ACT(\overleftarrow{P}))[i']) \text{ is an ancestor of } v \text{ in } \ACT(\overleftarrow{P})\}\\ 
& =  & \#(\acro{Leaves}_{\ACT(\overleftarrow{P})} (\acro{Parent}_{\ACT(\overleftarrow{P})} (\PA(\ACT(\overleftarrow{P}))[i']))).
\end{array}
  \]
Let $\{y_1,y_2,\ldots, y_{\#(T_B[i])}\}$ be the set of leaves of the subtree of $\acro{Parent}_{\ACT(\overleftarrow{P})} (\PA(\ACT(\overleftarrow{P}))[i'])$ in $\ACT(\overleftarrow{P})$ sorted such that $y_1 \lhd_{P.\sigma} \ldots \lhd_{P.\sigma} y_{\#(T_B[i])}$.
Given $k \in \inter{i''}{j''}$, we define $w_P[k]$ as the string $m_P[\Select_{\$}(m_P,\Rank_{\$}(m_P,\SA(P)[k] - 1 )) + 1 : SA(P)[k] + LRS(P)[k]-1]$. For all $k \in \inter{i''}{j''}$, the string $w_P[k]$ is a string of $P.S$. Moreover, the definition of $<_{P.\sigma}$ implies $w_P[i''] <_{P.\sigma} \ldots <_{P.\sigma} w_P[j'']$. As for $l \in \inter{1}{\#(T_B[i])}$, the string $\pi(y_l)$ is also a string of $P.S$, we get $\pi(y_l) = w_P[i''+l-1]$. 
Given $x$ and $y$ in $\inter{i''}{j''} \in  \Decomp(P)$, we obtain the following equivalences between orders:
 \[
      x < y \,\Leftrightarrow\,  w_P[x] <_{P.\sigma} w_P[y] \,\Leftrightarrow\, \pi(y_{x-i''+1}) <_{P.\sigma} \pi(y_{y-i''+1}) \,\Leftrightarrow\,   y_{x-i''+1} \lhd_{P.\sigma} y_{y-i''+1}.
  \]
By the definition of $\BWT(P)$, for all $k \in \inter{1}{\#(T_B[i])}$, we get that 
\[
\begin{array}{rcl}
\BWT(P)[i''+k-1] & = & m_P[\SA(m_P)[i''+k-1]-1]\\
& = & w_P[i''+k-1][\taille{w_P[i''+k-1]} - \LRS[i''+k-1] + 1] \\
& = & \pi(y_k)[\taille{w_P[i''+k-1]} - \LRS[i''+k-1] + 1]\\
& = & \delta[\acro{Child}_{\ACT(\overleftarrow{P})(y_k)} (\acro{Parent}_{\ACT(\overleftarrow{P})} (\PA(\ACT(\overleftarrow{P}))[i'])).
\end{array}
\]
\end{proof}

Theorem~\ref{th:link:bwt:xbw} provides us with a strong link between $\BWT(P)$ and $\XBWT(\ACT(\overleftarrow{P}))$, which allows transforming one a structure into the other. 
This leads to the following corollary.

\begin{corollary}\label{cor:build:bwt:xbw}
  \begin{itemize}
  \item Using tables $\BWT(P)$, $\LCP(P)$ and $\LRS(P)$ of an ordered set of strings $P$, we can build the tables $\XBWT(\ACT(\overleftarrow{P}))$ and $\XBWL(\ACT(\overleftarrow{P}))$ in linear time of $\norme{P.S}\times \#(\Sigma)$.
  \item Using tables $\XBWT(\ACT(S))$ and $\XBWL(\ACT(S))$ of a set of strings $S$, we can build the tables $\BWT(\overleftarrow{P})$, $\LCP(\overleftarrow{P})$ and $\LRS(\overleftarrow{P})$ in linear time of $\norme{S}\times \#(\Sigma)$ where 
  $P$ is an ordered set of strings 
  such that $P.S = S$.
  \end{itemize}
\end{corollary}

The idea behind the algorithms is to exploit the link of Theorem~\ref{th:link:bwt:xbw} to compute each substring of the \BWT or of the \XBWT associated to each element of $\Decomp$ or of $\Decompb$. The algorithms for computing the \XBWT and the proof of Corollary~\ref{cor:build:bwt:xbw} are given in Appendix.

\section{Optimal ordering of strings for maximising compression}\label{se:optimal}

\subsection{Minimum permutation problem for \BWT and \XBWT}

Run-Length Encoding~\cite{Siren09} is a widely used method to compress strings. For a string $w$, the Run-Length Encoding splits $w$ into the minimum number of substrings containing a single symbol. The size of the Run-Length Encoding of $w$ is the cardinality of the minimum decomposition.
For example for $abbaaaccabbb = a^1 b^2 a^3 c^2 a^1 b^3$ (using the power notation $\alpha^n$ means $n$ copies of symbol $\alpha$), the size of the Run-Length Encoding is $6$ (for the decomposition has $6$ blocks). 

We define Run-Length measures for a \BWT and for a \XBW (similar to those of~\cite{HoltM14b}). For $P$ an ordered set of strings, let $d_B(P)$ be the cardinality of the set $\{i \in \inter{1}{\#(\BWT(P)-1)} \ | \ \BWT(P)[i] \neq \BWT(P)[i+1] \}$. Similarly, let  $d_X(P)$ be the cardinality of the set $\{i \in \inter{1}{\#(\XBWT(\ACT(P))-1)} \ | \ \XBWT(\ACT(P))[i] \neq \XBWT(\ACT(P))[i+1] \}$. 

Given two ordered sets of strings, $P_1$ and $P_2$ such that $P_1.S = P_2.S$ (\ie, they contain the same set of strings), Theorem~\ref{th:link:bwt:xbw} implies that their $\BWT$ may differ, and thus $d_B(P_1)$ and $d_B(P_2)$ may also differ. We define the following minimisation problems.
As the Run-Length Encoding of $\BWT(P)$ has size $d_B(P)+1$, finding an optimal solution of \mpb can help compressing $\BWT(P)$.

\begin{definition}[\mpb and \mpx]
  Let $S$ be a set of strings. The problem \mpb asks for an ordered set of strings $P$ that minimises $d_B(P)$ and such that $P.S = S$. The problem \mpx aks for an ordered set of strings $P$ that minimises $d_X(P)$ and such that $P.S = S$. 
\end{definition}

To simplify \mpb, we consider specific ordered sets of strings.  Let $P$ be an ordered set of strings and let $\perp$ denote the root of $\ACT(\overleftarrow{P})$. We say that  is \emph{topologically planar} if for each node $u$ in $\ACT(\overleftarrow{P})$ and $v \in \acro{Leaves}_{\ACT(\overleftarrow{P})} (\perp) \setminus \acro{Leaves}_{\ACT(\overleftarrow{P})} (u)$, there does not exists $u_1$ and $u_2$ in $\acro{Leaves}_{\ACT(\overleftarrow{P})} (u)$ such that $ u_1 \lhd_{P.\sigma} v \lhd_{P.\sigma} u_2$. In other words, $P$ is \emph{topologically planar} if we can draw the tree $\ACT(\overleftarrow{P})$ by ordering the leaves with $\lhd_{P.\sigma}$ without arcs crossing each other.

Let $P$ be an ordered set of strings, which is not necessarily topologically planar. We denote by $P_{tp}$ the ordered set of strings such that $P_{tp}.S = P.S$ and $P_{tp}.\sigma$ such that for all $u$ in $\ACT(\overleftarrow{P})$, $v \in \acro{Leaves}_{\ACT(\overleftarrow{P})} (\perp) \setminus \acro{Leaves}_{\ACT(\overleftarrow{P})} (u)$ and $u_1$ and $u_2$ in $\acro{Leaves}_{\ACT(\overleftarrow{P})} (u)$ such that $ u_1 \lhd_{P.\sigma} v \lhd_{P.\sigma} u_2$, we have $ u_1 \lhd_{P_{tp}.\sigma} u_2 \lhd_{P_{tp}.\sigma} v$. As we have a bijection between the set of circular permutations of $P.S$ and the set of leaves of $\ACT(\overleftarrow{P})$, we can unambiguously define the ordered set of strings $P_{tp}$ that is topologically planar.

\begin{proposition}\label{prop:leq:topo:good}
  Let $P$ be an ordered set of strings. We have $d_B(P_{tp}) \leq d_B(P)$ and $d_B(P_{tp}) = d_X(\overleftarrow{P_{tp}})$.
\end{proposition}

\begin{proof}
  For the first inequality, let us prove that any modification of the order used to create $P_{tp}$ decreases the value of $d_B$. Let $P$ be an ordered set of strings which is not topologically planar. Let be $u$ in $\ACT(\overleftarrow{P})$, $v \in \acro{Leaves}_{\ACT(\overleftarrow{P})} (\perp) \setminus \acro{Leaves}_{\ACT(\overleftarrow{P})} (u)$ and $u_1$ and $u_2$ in $\acro{Leaves}_{\ACT(\overleftarrow{P})} (u)$ such that $ u_1 \lhd_{P.\sigma} v \lhd_{P.\sigma} u_2$. Let $P'$ the copy of $P$ where the only difference is $ u_1 \lhd_{P_{tp}.\sigma} u_2 \lhd_{P_{tp}.\sigma} v$. 
  Let be $x \in \Decomp	(P)$ such that $\acro{Dec\_Pre}[x] = u$.
  By Theorem~\ref{th:link:bwt:xbw}, for all $y \in \Decomp(P)\setminus \{x\}$, we have $\BWT(P)[y] = \BWT(P')[y]$, $\BWT(P)[x] = \ldots \delta[u_1]  \delta[v_1] \delta[u_2] \ldots$ and $\BWT(P)[x] = \ldots \delta[u_1]  \delta[u_2] \ldots \delta[v_1] \ldots $ with $v_1 = \acro{Child}_{\ACT(\overleftarrow{P})(v)} (u)$. As $\delta[u_1] = \delta[u_2]$ and $\delta[u_1] \neq \delta[v_1]$, we have  $d_B(P') \leq d_B(P)$.

  For the second inequality, it is enough to see that for an element $u$ in $\Decomp(P_{tp})$, the numbers of distinct successive symbols is identical in $\BWT(P_{tp})[u]$ and $\XBWT(\ACT(\overleftarrow{P_{tp}}))[\acro{BWT\_XBW}[u]]$. Thus, for two successive elements $u = \inter{i_1}{j_1}$ and $v = \inter{i_2}{j_2}$ of $\Decomp(P_{tp})$, we obtain an equivalence between $\acro{BWT\_XBW}[u] = \inter{i'_1}{j'_1}$ and $\acro{BWT\_XBW}[v] = \inter{i'_2}{j'_2}$:\\[.1cm]
 \centerline{$ \BWT(P_{tp})[j_1] = \BWT(P_{tp})[i_2] \quad \text{ \iif }\quad \BWT(\ACT(\overleftarrow{P_{tp}}))[j'_1] = \BWT(\ACT(\overleftarrow{P_{tp}}))[i'_2].$}
\end{proof}

Thanks to Proposition~\ref{prop:leq:topo:good}, we can restrict the search to ordered sets of strings that are topologically planar when solving \mpb or \mpx. Furthermore, an optimal solution of \mpb for $S$ is also an optimal solution of \mpx for $\overleftarrow{S}$, and vice versa. 
This yields the following theorem, whose proof is in Appendix.

\begin{theorem}\label{th:mpx:opt}
  Let $S$ be a set of strings. We can find an optimal solution for \mpb  and for \mpx in $O(\norme{S} \times \#(\Sigma))$ time.
\end{theorem}

\subsection{Proof of Theorem~\ref{th:mpx:opt}}

As a reminder, Proposition~\ref{prop:leq:topo:good} states that an optimal solution of \mpb is also an optimal solution of \mpx, and vice versa. In the following of this proof, we only prove the result regarding \mpx.

To start, let us give an overview of algorithm:
\begin{enumerate}
\item we take an random permutation $\sigma$ of $S$ and define $P$ such that $P.S = S$ and $P.\sigma= \sigma$ ,
\item we build $\ACT(P)$, $\XBWT(\ACT(P))$ and $\Decompb(\ACT(P))$,
\item\label{item:algo} we find $P'$ which is an optimal solution of \mpx.
\end{enumerate}
%

In the following, we define the problem \mpt and explicit its link to the problem \mpx (Lemma~\ref{le:equi:xbw:table}). Lemma~\ref{le:equi:xbw:table} gives us a linear algorithm for finding an optimal solution of \mpt, and thus we can apply this algorithm to obtain an optimal solution for \mpx.

Given $A$ an array of symbols of $\Sigma$, we define $\acro{Char}(A)$ as the set of (different) symbols in $A$.
Given $T$ an array of $n$ symbols of $\Sigma$ and $D$ an integer interval partition of $\inter{1}{n}$ such for each interval $\inter{i}{j}$ of $D$, $\acro{Char}(T[i,j]) = j-i+1$ (\ie all the symbols of $T[i,j]$ are different), the problem \mpt is to find a $T'$ such that for all $\inter{i}{j} \in D$, $\acro{Char}(T[i,j]) = \acro{Char}(T'[i,j])$ and which minimises $d_A(T',D) := \#\{i \in \inter{1}{n-1} \ | \ T[i] \neq T[i+1] \}$.

\begin{lemma}\label{le:equi:xbw:table}
  Let $S$ be a set of strings and let $P$ be an ordered set of strings such that $P.S = S$.
  For an optimal solution $T'$ of \mpt for $\XBWT(\ACT(P))$ and for $\Decompb(\ACT(P))$, there exists an optimal solution $P'$ of \mpx for $S$ such that $\XBWT(\ACT(P')) = T'$.
\end{lemma}

Let $T$ be an array of $n$ symbols of $\Sigma$ and let $D$ be an integer interval partition of $\inter{1}{n}$ such for each interval $\inter{i}{j}$ of $D$, $\acro{Char}(T[i,j]) = j-i+1$. Let $A(D)$ be the array of all intervals in $D$ in the order $<$ and $B(T,D)$ the array of size $\#(A(D))-1$ such that the position $i$ of $B(T,D)$ is $B(T,D)[i] = \acro{Char}(A(D)[i]) \cap \acro{Char}(A(D)[i+1])$. We define also $\acro{word}(C)$ for a set of symbols $C = \{c_1,\ldots,c_m\}$ the strings $c_1 \ldots c_m$ where $c_1 < \ldots < c_m$.

\begin{lemma}\label{le:decomp:mpt:all}
  Let $T$ be an array of $n$ symbols of $\Sigma$ and let $D$ be an integer interval partition of $\inter{1}{n}$ such for each interval $\inter{i}{j}$ of $D$, $\acro{Char}(T[i,j]) = j-i+1$. 
  \begin{itemize}
  \item If there exists $i$ in $\inter{1}{n}$ such that $\inter{i}{i}\in D$, we have $T'_1[1,i-1] T'_2$ is an optimal solution of \mpt for $T$ and for $D$ where $T'_1$ is an optimal solution of \mpt for $T[1,i]$ and for $\{\inter{i'}{j'} \in D \ | \ j\leq i\}$ and $T'_2$ in an optimal solution of \mpt for $T[i,n]$ and for $\{\inter{i'}{j'} \in D \ | \ i'\geq i\}$.
  \item If there exists $i$ in $\inter{1}{\#(B(T,D))}$ such that $\#(B(T,D)[i]) = 0$, we have $T'_1 T'_2$ is an optimal solution of \mpt for $T$ and for $D$ where $T'_1$ is an optimal solution of \mpt for $T[1,A(D)[i][1]]$ and for $\{\inter{i'}{j'} \in D \ | \ j\leq A(D)[i][1]\}$ and $T'_2$ in an optimal solution of \mpt for $T[A(D)[i+1][0],n]$ and for $\{\inter{i'}{j'} \in D \ | \ i'\geq A(D)[i+1][0]\}$.
  \item If there exists $i$ in $\inter{1}{\#(B(T,D))}$ such that $\#(B(T,D)[i]) = 1$, we have $T'_1 a a T'_2$ is an optimal solution of \mpt for $T$ and for $D$ where $B(T,D)[i] = \{a\}$, $T'_1$ is an optimal solution of \mpt for $T[1,A(D)[i-1][1]]\acro{word}(\acro{Char}(A(D)[i]\setminus \{a\})$ and for $\{\inter{i'}{j'} \in D \ | \ j < A(D)[i][1]\}\cup\{\inter{A(D)[i][0]}{A(D)[i][1]-1} \}$ and $T'_2$ in an optimal solution of \mpt for $\acro{word}(\acro{Char}(A(D)[i+1]\setminus \{a\})T[A(D)[i+2][0],n]$ and for $\{\inter{i'}{j'} \in D \ | \ i'> A(D)[i+1][0]\}\cup\{\inter{A(D)[i+1][0]+1}{A(D)[i][1]} \}$.
  \end{itemize}
\end{lemma}

\begin{proof}
  All the proofs are derived from the equality $d_A(T[1,n],D) = d_A(T[1,i],\{\inter{i'}{j'} \in D \ | \ j\leq i\}) + d_A(T[i,n],\{\inter{i'}{j'} \in D \ | \ i'\geq i\})$ for all $ i \in \inter{1}{n}$.
\end{proof}

\begin{lemma}\label{le:opt:mpt}
  Let $T$ be an array of $n$ symbols of $\Sigma$ and let $D$ be an integer interval partition of $\inter{1}{n}$ such for each interval $\inter{i}{j}$ of $D$, $\acro{Char}(T[i,j]) = j-i+1$. In the case where for all $i \in \inter{1}{\#(B(T,D))}$, $\#(B(T,D)[i]) \geq 2$, Algorithm~\ref{algo:mpt} gives an optimal solution of \mpt in $\#(T)\times \#(\Sigma)$.
\end{lemma}

\begin{algorithm}[ht]
  \SetKwInOut{Input}{Input}
  \SetKwInOut{Output}{Output}
  \Input{An instance of \mpt $T$ and $D$}
  \Output{A string $T'$}

  $last \leftarrow  \$$ such that $\$ \notin \Sigma$\;
  $T' \leftarrow $ empty string\;

  \For{$i \in \inter{1}{\#(B(T,D))}$}{
    $lettre \leftarrow $ random element of  $B(T,D)[i] \setminus \{last\}$\;
    $T' \leftarrow T'\; \acro{word}(\acro{Char}(A(D)[i])\setminus \{lettre,last\})$\;
    $T' \leftarrow T'\;  lettre$\;
    $T' \leftarrow T'\;  lettre$\;
    $last \leftarrow lettre$\;
  }
  $T' \leftarrow T'\; \acro{word}(\acro{Char}(A(D)[\#(A(D))])\setminus \{last\})$\;
  \Return $T'$\;

  \caption{Computation of an array $T'$ of symbols satisfying for any $\inter{i}{j} \in D$, $\acro{Char}(T[i,j]) = \acro{Char}(T'[i,j])$.\label{algo:mpt}}
\end{algorithm}

\begin{proof}
  \textbf{Complexity} To build the tables $A(D)$ and $B(T,D)$, we need $O(\#(T)\times \#(\Sigma))$ in time. As the size of these two tables are smaller than the size of $T$, the loop \textbf{for} of Algorithm~\ref{algo:mpt} takes also $O(\#(T)\times \#(\Sigma))$ in time.

  \textbf{Optimality} 
  As for all $i$ in $\inter{1}{\#(B(T,D))}$, $\#(B(T,D)[i]) \geq 2$, we have $B(T,D)[i] \setminus \{last\} \neq \emptyset$.
  Let $T^{\ast}$ be a string such that for all $\inter{i}{j} \in D$, $\acro{Char}(T[i,j]) = \acro{Char}(T^{\ast}[i,j])$. The size of $T^{\ast}$ is $\#(T)$ and the number of interval of $D$ is $\#(D)$, \ie\ the maximum number of positions where two consecutive letters can be identical. Hence, we have $d_A(T^{\ast},D) \leq \#(T) - \#(D) + 1$. Let $T'$ be the string given by Algorithm~\ref{algo:mpt}. We have $d_A(T',D) = \sum_{u \in D} (\#(u)-1) + 1 = \#(T) - \#(D) + 1$. This concludes the proof.
\end{proof}

\begin{lemma}\label{le:opt:mpt:all}
  Let $T$ be an array of $n$ symbols of $\Sigma$ and let $D$ be an integer interval partition of $\inter{1}{n}$ such that for each interval $\inter{i}{j}$ of $D$, $\acro{Char}(T[i,j]) = j-i+1$. The problem \mpt can be solved in linear time in $\#(T)\times \#(\Sigma)$.
\end{lemma}

\begin{proof}
  By Lemma~\ref{le:opt:mpt} and Lemma~\ref{le:decomp:mpt:all}, we can compute an optimal solution of \mpt by cuting the interval, applying Algorithm~\ref{algo:mpt} on each part, and then merge the strings output by Algorithm~\ref{algo:mpt}.
\end{proof}

\section{Conclusion and Perspectives}


{
Here, we present a new view of the Burrows-Wheeler Transform: as the text representation of an Aho-Corasick automaton that depends on the concatenation order. This induces a link between the Burrows-Wheeler Transform and the eXtended Burrows-Wheeler Transform, via the Aho-Corasick automaton. This link allows one to transform one structure into the other (for which we provide algorithms). We also exploit this link to find in linear time an ordering of input strings that optimises the compression of the concatenated strings.\footnote{An implementation of these algorithms can be found in \url{https://framagit.org/bcazaux/compressbwt}}
}
\bibliography{esa}

\newpage
\section*{Appendix}

\section*{Proof of Lemma~\ref{le:lrs}}

As for all $i \in \inter{1}{\taille{w}}$, $\SA(w)[\LF(w)[i]] = \SA(w)[i]-1$ if $\SA(w)[i] \geq 2$ and $\SA(w)[\LF(w)[i]] = \taille{w}$ otherwise, we have
\begin{itemize}
\item if $\SA(w)[i] \geq 2$,
  \[
    \begin{array}{rcl}
      \LRS(P)[\LF(P)[i]] & = & \Select_{\$}(m_P[\SA(m_P)[\LF(P)[i]]:\taille{m_P}],1) - 1 \\
                         & = & \Select_{\$}(m_P[\SA(w)[i]-1:\taille{m_P}],1) - 1 \\
                         & = & 
                               \begin{cases}
                                 0 & \text{if } m_P[\SA(w)[i]-1] = \$, \\
                                 \LRS(P)[i]+1 & \text{otherwise}.
                               \end{cases}
    \end{array}
  \]
\item If $\SA(w)[i] = 1$,
  \[
    \begin{array}{rcl}
      \LRS(P)[\LF(P)[i]] & = & \Select_{\$}(m_P[\SA(m_P)[\LF(P)[i]]:\taille{m_P}],1) - 1 \\
                         & = & \Select_{\$}(m_P[\taille{m_P}:\taille{m_P}],1) - 1 \\
                         & = & 0.
    \end{array}
  \]
\end{itemize}

We define the function $\acro{Letter}$ from $\inter{1}{\taille{m_P}}$ to $\Sigma$ such that $C[\acro{Letter}[i]] < i \leq C[\acro{Letter}[i+1]]$ (see definition of \LF) where $\Sigma = \{c_1 \ldots, c_{\#(\Sigma)} \}$ with $c_1 < \ldots < c_{\#(\Sigma)}$. 
Thus, we define $\acro{RLF}(P)$ from $\inter{1}{\taille{m_P}}$ to $\inter{1}{\taille{m_P}}$ such that 
\[
  \acro{RLF}(P)[i] = \Select_{\acro{Letter}[i]}(\BWT(P),i-C[\acro{Letter}[i]]).
\]
Let $i$ be an integer between $1$ and $\taille{m_P}$. We have
\[
  \begin{array}{rcl}
    \acro{RLF}(P)[\LF(P)[i]] & = & \Select_{\acro{Letter}[\LF(P)[i]]}(\BWT(P),\LF(P)[i]-C[\acro{Letter}[\LF(P)[i]]])\\ 
                             & = & \Select_{\BWT(P)[i]}(\BWT(P),\LF(P)[i]-C[\BWT(P)[i]])\\
                             & = & \Select_{\BWT(P)[i]}(\BWT(P),\Rank_{\BWT(w)[i]}(\BWT(w),i))\\ 
                             & = & i. 
  \end{array}
\]
Hence, the function $\acro{RLF}(P)$ is the reverse bijection of $\LF(P)$, and as $m_P[\SA(m_P)[i]-1] = \BWT(P)[i]$, one gets
\[
  \begin{array}{rcc}
    LRS(P)[i]  = 0 & \Leftrightarrow & \BWT(P)[\acro{RLF}(P)[i]] = \$ \\
                   & \Leftrightarrow & \BWT(P)[\Select_{\acro{Letter}[i]}(\BWT(P),i-C[\acro{Letter}[i]])] = \$ \\
                   & \Leftrightarrow & \acro{Letter}[i] = \$ \\
                   & \Leftrightarrow & i \in \inter{1}{\#(P.S)}\\
  \end{array}
\]
Therefore, we derive the following equality
\[
  \LRS(P)[i] = 
  \begin{cases}
    0 & \text{if } i \in \inter{1}{\#(P.S)},\\
    \LRS(P)[\LF(P)[i]]-1 & \text{otherwise}.
  \end{cases}
\]

\section*{Proof of Proposition~\ref{prop:linear:lrs}}

As the value of each position of the table $\LCP(P)$ corresponds to the minimum between the values of same position of $\LCP(m_P)$ and of $\LRS(P)$, we only need to proove that the table $\LRS(P)$ can be computed in linear time from $\BWT(P)$.

Using Algorithm~\ref{algo:lrs} and Lemma~\ref{le:lrs}, we can compute table $\LRS(P)$ in $O(\taille{m_P})$ time.

\begin{algorithm}[ht]
  \SetKwInOut{Input}{Input}
  \SetKwInOut{Output}{Output}
  \Input{The ordered set of strings $P$}
  \Output{The table $\LRS(P)$}

  We compute the table $\BWT(P)$\;
  We build the empty table $LRS$ of length $\taille{m_P}$\;

  \For{$i \in \inter{1}{\#(P.S)}$}{
    $position \leftarrow i$\;
    $nb \leftarrow 0$\;
    $LRS[position] \leftarrow nb$\;

    \While{$\BWT(P)[position] \neq \$ $}{
      $LRS[position] \leftarrow nb$\;
      $nb \leftarrow nb+1$\;
      $position \leftarrow \LF[position]$\;	
    }
  }
  \Return $LRS$\;

  \caption{Computation of table $\LRS(P)$.\label{algo:lrs}}
\end{algorithm}

\subsection*{Proof of Proposition~\ref{prop:bij:decomp:prefix}}

We begin by giving the following Lemma.
\begin{lemma}\label{le:value:same:decomp}
  Let $u =\inter{i}{j} \in \Decomp(P)$. For all $k \in \inter{i}{j}$, \[\acro{Dec\_Pre}[u] \; = \; \overleftarrow{m_P[\SA(m_P)[k]:\SA(m_P)[k]+\LRS(P)[k]-1]}.\]
\end{lemma}

\begin{proof}
  Let us show by contraposition that $\LRS(P)[k-1] = \LCP(P)[k]$ for all $k \in \inter{i+1}{j}$. Assume that there exists $k \in \inter{i+1}{j}$ such that $\LRS(P)[k-1] \neq \LCP(P)[k]$. 
  Whenever $\LRS(P)[k-1] < \LCP(P)[k]$, we get by definition that $\LRS(P)[k] \geq \LCP(P)[k]$, and thus $\LRS(P)[k-1] < \LRS(P)[k]$. By the definiton of $\LRS(P)$, $m_P[\SA(m_P)[k-1]+\LRS(P)[k-1]] = \$ $. By the definition of $\LCP(P)$, for all $j \in \inter{1}{\LCP(P)[k-1]-1}$, $m_P[\SA(m_P)[k]+j] = m_P[\SA(m_P)[k-1]+j]$. As $\LRS(P)[k-1] < \LCP(P)[k]$, we have $m_P[\SA(m_P)[k]+\LRS(P)[k-1]] = m_P[\SA(m_P)[k-1]+\LRS(P)[k-1]] = \$ $, which is impossible since $\LRS(P)[k-1] < \LRS(P)[k]$.
  Whenever $\LRS(P)[k-1] > \LCP(P)[k]$, as $\LCP(P)[k] = \LRS(P)[k]$, the string $m_P[\SA(m_P)[k]:\taille{m_P}]$ is lexicographically strictly smaller than $m_P[\SA(m_P)[k-1]:\taille{m_P}]$, which is impossible.
  This concludes the proof.
\end{proof}

  Let $u = \inter{i}{j}$ be an interval of $\Decomp(P)$. First, we prove that $\acro{Dec\_Pre}[u] \in \Prefix(\overleftarrow{P.S})$, and then to prove the bijection, we show $\acro{Dec\_Pre}$ is injective and surjective. By definition, $\LRS(P)[i] = \Select_{\$}(m_P[\SA(m_P)[i]:\taille{m_P}],1) - 1$, we have $m_P[\SA(m_P)[i]+\LRS(P)[i]]= \$ $ and for all $j \in \inter{\SA(m_P)[i]}{\SA(m_P)[i]+\LRS(P)[i]-1}$, $m_P[j] \neq \$ $. Thus, $m_P[\SA(m_P)[i]:\SA(m_P)[i]+\LRS(P)[i]-1]$ is a suffix of a string $w$ of $P.S$ and thus $\acro{Dec\_Pre}[u]$ is a prefix of $\overleftarrow{w}$ in $\overleftarrow{P.S}$.

  Let $u_1 = \inter{i_1}{j_1}$ and $u_2 = \inter{i_2}{j_2}$ be two elements of $\Decomp(P)$. Without loosing generality, we take $i_1 \leq i_2$.
  Assume that $\acro{Dec\_Pre}[u_1] = \acro{Dec\_Pre}[u_2]$, we have that 
  $m_P[\SA(m_P)[i_1]:\SA(m_P)[i_1]+\LRS(P)[i_1]-1] = m_P[\SA(m_P)[i_2]:\SA(m_P)[i_2]+\LRS(P)[i_2]-1]$ and thus for all $k \in \inter{i_1}{i_2}$, $m_P[\SA(m_P)[i_1]:\SA(m_P)[i_1]+\LRS(P)[i_1]-1] = m_P[\SA(m_P)[k]:\SA(m_P)[k]+\LRS(P)[k]-1]$. Hence, we have $\LCP(P)[k] = \LRS(P)[i_1]$ and $\LRS(P)[k] = \LRS(P)[i_1]$. Therefore by the definition of $\Decomp(P)$, we get $u_1 = u_2$. 

  Let $v$ be a prefix of a string of $\overleftarrow{P.S}$. The string $\overleftarrow{v}$ is a suffix of a string of $P.S$. By the definition of $m_P$, $\overleftarrow{v}$ is a prefix of a suffix $s$ of $m_P$ such that $s[\taille{\overleftarrow{v}}+1] = \$ $. By the definition of $\BWT(P)$, the table $\SA(m_P)$ gives for a position $i$ the starting position of the $i^{th}$ suffix of $m_P$ in lexicographic order. Hence, there is a bijection between  $\Suffix(m_P)$ and the set of positions in $\SA(m_P)$. Let $k \in \inter{1}{\taille{m_P}}$ such that $s =  m_P[\SA(m_P)[k]:\taille{m_P}]$. As $\overleftarrow{v}$ is a suffix of a string of $P.S$ and a prefix of $s$, $\overleftarrow{v} =  m_P[\SA(m_P)[k]:\SA(m_P)[k]+\LRS(P)[k]-1]$. We take $u = \inter{i}{j} \in \Decomp(P)$ such that $k \in \inter{i}{j}$. By Lemma~\ref{le:value:same:decomp}, $\overleftarrow{v} = \acro{Dec\_Pre}[u]$.

\subsection*{Proof of Proposition~\ref{prop:iso:act}}

  First, we show that there exists a bijection between the node set of $\ACT(\overleftarrow{P.S})$ and that of $G_T(P)$. We reuse the bijection $\acro{Dec\_Pre}$, which served in Proposition~\ref{prop:bij:decomp:prefix}.
  Let us show that for each arc $(u,v)$ of $A_T(P)$, $(\acro{Dec\_Pre}[u],\acro{Dec\_Pre}[v])$ is an arc of $\ACT(\overleftarrow{P.S})$, and vice versa.
  Let be $(u,v) \in A_T(P)$, i.e. $(u,v) \in \Decomp(P)^2$ such that there exists $x \in u$ with $\LF(P)[x] \in v$. 


  By the work of~\cite{Burrows94}, we know that $\SA(m_P)[\LF(m_P)[i]] = \SA(m_P)[i]-1$. By Lemma~\ref{le:lrs}, we have $\LRS(P)[\LF(P)[x]] = \LRS(P)[x] + 1$ for all $x$ such that $\LF(P)[x] \neq 0$. With both equalities and Lemma~\ref{le:value:same:decomp}, we obtain
  \[
    \begin{array}{rcl}
      \acro{Dec\_Pre}[v] & = & \overleftarrow{m_P[\SA(m_P)[\LF(P)[x]]:\SA(m_P)[\LF(P)[x]]+\LRS(P)[\LF(P)[x]]-1]} \\
                         & = & \overleftarrow{m_P[\SA(m_P)[x]-1:\SA(m_P)[x]-1+\LRS(P)[\LF(P)[x]]-1]} \\
                         & = & \overleftarrow{m_P[\SA(m_P)[x]-1:\SA(m_P)[x]-1+\LRS(P)[x]]} \\
                         & = & \overleftarrow{m_P[\SA(m_P)[x]:\SA(m_P)[x]+\LRS(P)[x]-1]}\ m_P[\SA(m_P)[x]-1]\\
                         & = & \acro{Dec\_Pre}[u]\ m_P[\SA(m_P)[x]-1].
    \end{array}
  \]
  The string $\acro{Dec\_Pre}[u]$ is thus the longest prefix of $\acro{Dec\_Pre}[v]$.

  Let $(x,y)$ be an arc of $\ACT(\overleftarrow{P.S})$. We take $z$ a leaf in the subtree of $\ACT(\overleftarrow{P.S})$ in $y$. As $x$ is the parent of $y$ in $\ACT(\overleftarrow{P.S})$, $z$ is also a leaf in the subtree of $\ACT(\overleftarrow{P.S})$ in $x$. We take $i \in \inter{1}{\taille{m_P}}$ such that $\overleftarrow{z}$ is a prefix of $m_P[i:\taille{m_P}]$. Hence $\overleftarrow{y}$ is a prefix of $m_P[i+\taille{z}-\taille{y}:\taille{m_P}]$ and $\overleftarrow{x}$ is a prefix of $m_P[i+\taille{z}-\taille{x}:\taille{m_P}]$. As $(x,y)$ is an arc of $\ACT(\overleftarrow{P.S})$, $\taille{y} - \taille{x} = 1$. Thus by choosing $k$ such that $\SA(m_P)[k] = i+\taille{z}-\taille{y}+1$, and $u,v$ in $\Decomp(P)^2$ such that $k\in u$ and $\LF(P)[k] \in v$, we get $\overleftarrow{x} = \acro{Dec\_Pre}[u]$ and $\overleftarrow{y} = \acro{Dec\_Pre}[v]$.
  This concludes the proof.

\subsection*{Proof of Proposition~\ref{prop:iso:acfl}}

  First, let us show the following equivalence. Let be $u = \inter{i}{j} \in \Decomp(P)$.
  \[
    \begin{array}{c}
      w \in \Decomp(P) \text{ such that } \exists k \in w \text{ with } k<i \text{ and } \LRS(P)[k] = \min_{k-1 \leq l \leq i}(\LCP(P)[l]) \\
      \Leftrightarrow \\
      \acro{Dec\_Pre}[w] \text{ is a suffix of } \acro{Dec\_Pre}[u].
    \end{array}
  \]
  Let be $w \in \Decomp(P)$  such that $\exists k \in w \text{ with } k<i \text{ and } \LRS(P)[k] = \min_{k-1 \leq l \leq i}(\LCP(P)[l])$.
  Hence, we have for all $l \in \inter {k-1}{i}$, $\LRS(P)[k] \leq \LCP(P)[l]$, and thus 
  \[
    \begin{array}{rcl}
      \acro{Dec\_Pre}[u] & = & \overleftarrow{m_P[\SA(m_P)[i]:\SA(m_P)[i]+\LRS(P)[i]-1]}\\
                         & = & \overleftarrow{m_P[\SA(m_P)[i]+\LRS(P)[k]:\SA(m_P)[i]+\LRS(P)[i]-1]}\\
                         & & \overleftarrow{m_P[\SA(m_P)[i]:\SA(m_P)[i]+\LRS(P)[k]-1]}\\
                         & = & \overleftarrow{m_P[\SA(m_P)[i]+\LRS(P)[k]:\SA(m_P)[i]+\LRS(P)[i]-1]}\\
                         & & \overleftarrow{m_P[\SA(m_P)[k]:\SA(m_P)[k]+\LRS(P)[k]-1]}\\
                         & = & \overleftarrow{m_P[\SA(m_P)[i]+\LRS(P)[k]:\SA(m_P)[i]+\LRS(P)[i]-1]}\ \acro{Dec\_Pre}[w].
    \end{array}
  \]
  Let $w = \inter{i_1}{j_1}$ and $u = \inter{i_2}{j_2}$ be two elements of $\Decomp(P)$ such that $\acro{Dec\_Pre}[w]$ is a suffix of $\acro{Dec\_Pre}[u]$. Hence, we have that $m_P[\SA(m_P)[i_1]:\SA(m_P)[i_1]+\LRS(P)[i_1]-1]$ is a prefix of $m_P[\SA(m_P)[i_2]:\SA(m_P)[i_2]+\LRS(P)[i_2]-1]$. By the definition of $\BWT(P)$, for all $l \in \inter {i_1}{i_2}$, $\LRS(P)[i_1] \leq \LCP(P)[l]$. This concludes the proof of the equivalence.
  By the equivalence, given $u$ and $v$ in $\Decomp(P)$ such that $\acro{Dec\_Pre}[u]$ is a suffix of $\acro{Dec\_Pre}[v]$, for all $k_1 \in u$ and $k_2 \in v$, we have $k_1 \leq k_2$. Hence, by taking the largest $w$ satisfying the first step of the inequality, we obtain the longest suffix and vice versa.

\subsection*{Proof of Corollary~\ref{cor:build:bwt:xbw}}

\subsubsection*{From \BWT to \XBW}

Let $P$ be an ordered set of strings. To compute tables $\XBWT(\ACT(\overleftarrow{P}))$ and $\XBWL(\ACT(\overleftarrow{P}))$ using only $\BWT(P)$, $\LCP(P)$ and $\LRS(P)$, we first define a new table $\BWD(P)$.

The \emph{Burrows-Wheeler Decomposition} of $P$, denoted by $\BWD(P)$, is the array of length $\#(\Decomp(P))$ such that for each position $i$, $\BWD(P)[i]$ is the cardinality of the $i^{th}$ element of $\Decomp(P)$ in interval order.

\begin{lemma}\label{le:build:bwd}
Using tables $\LCP(P)$ and $\LRS(P)$, Algorithm~\ref{algo:bwd} computes $\BWD(P)$ in linear time in $\norme{P.S}$ and the table $\BWD(P)$ can be stored with $\norme{P.S} \times log(\#(P.S))$ bits.
\end{lemma}

\begin{proof}[Proof of Lemma~\ref{le:build:bwd}]
For each $i$ in $\inter{1}{\taille{m_P}}$, at the begining of the loop \textbf{for}, we have that $\LCP(P)[i-j] \neq \LRS(P)[i-j]$ and for all $k \in \inter{i-j+1}{i-1}$, $\LCP(P)[k] \neq \LRS(P)[k]$. Hence, if $\LCP(P)[i] \neq \LRS(P)[i]$, the interval $\inter{i-j}{i-1}$ is an element of $\Decomp(P)$ and the cardinality of $\inter{i-j}{i-1}$ is $j$. Otherwise, we increase $j$ by $1$ because the position $i$ does not correspond to a new interval of $\Decomp(P)$.
For the complexity, as each step of the loop can be computed in constant time, Algorithm~\ref{algo:bwd} computes $\BWD(P)$ in linear time in $\norme{P.S}$.
As for each position $i$ of $\BWD(P)$, $\BWD(P)[i]$ represents the number of strings of $P.S$ having as suffix $\acro{Dec\_Pre}[u]$, where $u$ is the $i^{th}$ element of $\Decomp(P)$ sorted in interval order, it follows that $\BWD(P)[i] \leq \#(P.S)$. This concludes the proof.
\end{proof}

\begin{algorithm}[ht]
  \SetKwInOut{Input}{Input}
  \SetKwInOut{Output}{Output}
  \Input{The tables $\LCP(P)$ and $\LRS(P)$}
  \Output{The table $\BWD(P)$}

	$BWD \leftarrow $ empty list\;

    $j \leftarrow 0$\;
  \For{$i \in \inter{1}{\taille{m_P}}$}{
	\uIf{$\LCP(P)[i] \neq \LRS(P)[i]$}{
   		add $j$ to the end of $BWD$ \;
   		$j \leftarrow 1$\;
  	}
  	\Else{
    	$j \leftarrow j+1$\;
  	} 
  }
  \Return $BWD$\;

  \caption{Computation of table $\BWD(P)$.\label{algo:bwd}}
\end{algorithm}

\begin{lemma}\label{le:go:bwd:xbwd}
Using tables $\BWT(P)$ and $\BWD(P)$, Algorithm~\ref{algo:bwt:xbw} computes the tables $\XBWT(\ACT(\overleftarrow{P}))$ and $\XBWL(\ACT(\overleftarrow{P}))$ in linear time of $\norme{P.S} \times \#(\Sigma)$.
\end{lemma}

\begin{proof}[Proof of Lemma~\ref{le:go:bwd:xbwd}]
Algorithm~\ref{algo:bwt:xbw} is an application of Theorem~\ref{th:link:bwt:xbw}.
\end{proof}

\begin{algorithm}[ht]
  \SetKwInOut{Input}{Input}
  \SetKwInOut{Output}{Output}
  \Input{The tables $\BWT(P)$ and $\BWD(P)$}
  \Output{The tables $\XBWT(\ACT(\overleftarrow{P}))$ and $\XBWL(\ACT(\overleftarrow{P}))$}

	$XBWT \leftarrow $ empty list\;
	$XBWL \leftarrow $ empty list\;

	$nb \leftarrow 1$\;

  \For{$i \in \inter{1}{\#(\BWD(P))}$}{
		$D \leftarrow$ Dictionnary such that for all $l \in \Sigma$, $D[l] \leftarrow true$\;
		 $begin \leftarrow nb$\;
		 $last \leftarrow begin+\BWD(P)[i]-1$\;
		 $nb \leftarrow last+1$\;

		\For{$j \in \inter{begin}{last}$}{
			\uIf{$D[\BWT(P)[j]]$}{
				add $\BWT(P)[j]$ to the end of $XBWT$ \;
   				add $0$ to the end of $XBWL$ \;
   				$D[\BWT(P)[j]] \leftarrow false$\;
  			}		
		}
		$XBWL[\#(XBWL)] \leftarrow 1$\;
		
  }
  \Return $XBWT$ and $XBWL$\;

  \caption{Computation of tables $\XBWT(\ACT(\protect\overleftarrow{P}))$ and $\XBWL(\ACT(\protect\overleftarrow{P}))$.\label{algo:bwt:xbw}}
\end{algorithm}

Using Algorithm~\ref{algo:bwd} to build $\BWD(P)$ and Algorithm~\ref{algo:bwt:xbw}, we can compute tables $\XBWT(\ACT(\overleftarrow{P}))$ and $\XBWL(\ACT(\overleftarrow{P}))$ in linear time of $\norme{P.S} \times \#(\Sigma)$.

\subsubsection*{From \XBW to \BWT}

We define the equivalent of \BWD for $\Decompb(P)$. The \emph{eXtended Burrows Wheeler Decomposition} (\XBWD) of a tree $\mathcal{T}$ is the array of length of $\XBWL(\mathcal{T})$ such that for each position $i$, $\XBWD(P)[i]$ equals the cardinality of the $i^{th}$ element of $\Decompb(\mathcal{T})$ sorted in interval order. 

\begin{lemma}\label{le:build:xbwd}
  From the table $\XBWL(\ACT(S))$, Algorithm~\ref{algo:xbwd} computes $\XBWD(\ACT(S))$ in linear time in $\norme{P.S}$ and the table $\XBWD(\ACT(S))$ can be stored in $\norme{P.S} \times log(\#(\Sigma))$ bits.
\end{lemma}

\begin{proof}
The proof of Lemma~\ref{le:build:xbwd} is similar to the proof of Lemma~\ref{le:build:bwd}. As for each position $i$ of $\XBWD(\ACT(S))$, $\XBWD(\ACT(S))[i]$ represents the number of the right extension of the strings $\PA(\ACT(S))[i]$ in $S$ (i.e. the number of different strings $\PA(\ACT(S))[i]\; a$ which are substrings of a string of $S$ with $a \in \Sigma$), we have $\XBWD(\ACT(S))[i] \leq \#(\Sigma)$. 
\end{proof}

\begin{algorithm}[ht]
  \SetKwInOut{Input}{Input}
  \SetKwInOut{Output}{Output}
  \Input{The table $\XBWL(\ACT(S))$}
  \Output{The table $\XBWD(\ACT(S))$}

	$XBWD \leftarrow $ empty list\;

    $j \leftarrow 1$\;
  \For{$i \in \inter{1}{\#(\XBWL(\ACT(S)))}$}{
	\uIf{$\XBWL(\ACT(S)) = 1$}{
   		add $j$ to the end of $XBWD$ \;
   		$j \leftarrow 1$\;
  	}
  	\Else{
    	$j \leftarrow j+1$\;
  	} 
  }
  \Return $XBWD$\;

  \caption{Computation of table $\XBWD(\ACT(S))$.\label{algo:xbwd}}
\end{algorithm}

\begin{lemma}
  Using tables $\XBWT(\ACT(S))$ and $\XBWD(\ACT(S))$, we can build the tables $\BWT(\overleftarrow{P})$, $\LCP(\overleftarrow{P})$ and $\LRS(\overleftarrow{P})$ in linear time of $\norme{P.S}\times \#(\Sigma)$ where $P$ is a  topologically planar, ordered set of strings such that $P.S = S$.
\end{lemma}

\begin{proof}
  In~\cite{FerraginaLMM09}, Ferragina \etal  prove that with both tables $\XBWT(\ACT(S))$ and $\XBWD(\ACT(S))$ one can access in constant time the children and the parent in $\ACT(S)$. Hence, we can compute in linear time in $\norme{P.S}$, the table $\acro{TL}(\ACT(S))$, where in each position of $i$ we store the number of leaves in the subtree of the node $\PA(\ACT(S))[i]$. We finish the proof using the results of Theorem~\ref{th:link:bwt:xbw} and an algoritm similar to Algorithm~\ref{algo:bwt:xbw}.
\end{proof}


\subsection*{Proof of Lemma~\ref{le:equi:xbw:table}}

  Let $T'$ be an optimal solution of \mpt for $\XBWT(\ACT(P))$ and for $\Decompb(\ACT(P))$. By Theorem~\ref{th:link:bwt:xbw}, for each $\inter{i}{j} \in \Decompb(\ACT(P))$, the order of the symbols in $\XBWT(\ACT(P))[i,j]$ depends of the order on the children of the parent of $\PA(\ACT(P))[i]$. Hence, the choice of $T'$ corresponds to the choice of an order for each internal node of $\ACT(P)$ over all its children. As we can extend this type of order to an total order over the leaves of $\ACT(P)$, we can build $P'$ the ordered set of strings such that $P'.S = S^{\$}$ and $P'.\sigma(\pi(f_i)) = s_{i}$ with the order over the leaves of $\ACT(P)$ gives $f_1 < \ldots < f_{\#(S)}$ and $s_1 < \ldots < s_{\#(S)}$ are the strings of $S$ in lexicographic order.

\end{document}